\documentclass[11pt,a4paper]{article}
\usepackage{fullpage,amssymb,amsmath,amsfonts,amsthm,hyperref,graphicx}

\newtheorem {theorem} {Theorem}

\begin{document}
\title{Hyperbolae are the locus of constant angle difference}
\author{Herman Haverkort and Rolf Klein\\Department of Computer Science, University of Bonn}
\date{1 December 2021}
\maketitle

\begin{abstract}
\noindent
Given two points $A,B$ in the plane, the locus of all points $P$ for which the angles at $A$ and $B$ in the triangle
$A,B,P$ have a constant \emph{sum} is a circular arc, by Thales' theorem.  
We show that the \emph{difference} of these angles is kept a constant by points $P$ on a hyperbola
(albeit with foci different from $A$ and $B$).
Whereas hyperbolae are well-known to maintain a constant difference between the distances to their foci, the above angle property seems not to be widely known. The question was motivated by recent work by Alegr\protect{\'{i}}a et al. \cite{ampssss-vdrra-21}  and De Berg et al.~\cite{bghh-vdrdc-17} on Voronoi diagrams of turning rays.

\addvspace\baselineskip
\end{abstract}

\section{Rotating ray Voronoi diagrams}\label{intro-sect}
Quite recently, Alegr\protect{\'{i}}a et al. \cite{ampssss-vdrra-21} presented a novel type of Voronoi diagrams, where each point site $p$ is associated with 
a fixed ray $\phi(p)$ emanating from $p$, and the distance from $p$ to some point $z\not=p$ in the plane is defined as the \emph{counterclockwise} angle between $\phi(p)$ and the ray from $p$ to $z$. In general, the curve separating the Voronoi regions of two sites $p_1,p_2$ consists of a circular arc connecting the sites, plus the 
two rays $\phi(p_1)$ and $\phi(p_2)$. However, the points on ray $\phi(p_1)$ have distance zero from $p_1$ but a positive distance  from $p_2$ (unless $p_2 \in \phi(p_1)$), and the same holds for $p_2$.

This contrasts with the ``angle-only'' version of the Voronoi diagrams with rotational distance costs as proposed by De Berg et al.~\cite{bghh-vdrdc-17}. In these diagrams, the distance from $p$ to some point $z\not=p$ is defined as the 
the \emph{minimum} of the clockwise and the  counterclockwise angles between $\phi(p)$ and the ray $\vec{pz}$.
By this definition, the distance measure becomes continuous even along the rays.  
Now the locus of all points of identical distance from $p_1$ and $p_2$ does separate the Voronoi regions of two sites. 
Due to the minimization in the distance definition, regions may be disconnected; see Figure \ref{sym-fig}. Archimedian spirals  with polar coordinates $(\alpha,\alpha)$ around their respective centers at $p$ and $q$ help to visualize angles as lengths.
\begin{figure}
\begin{center}
\includegraphics[scale=0.6]{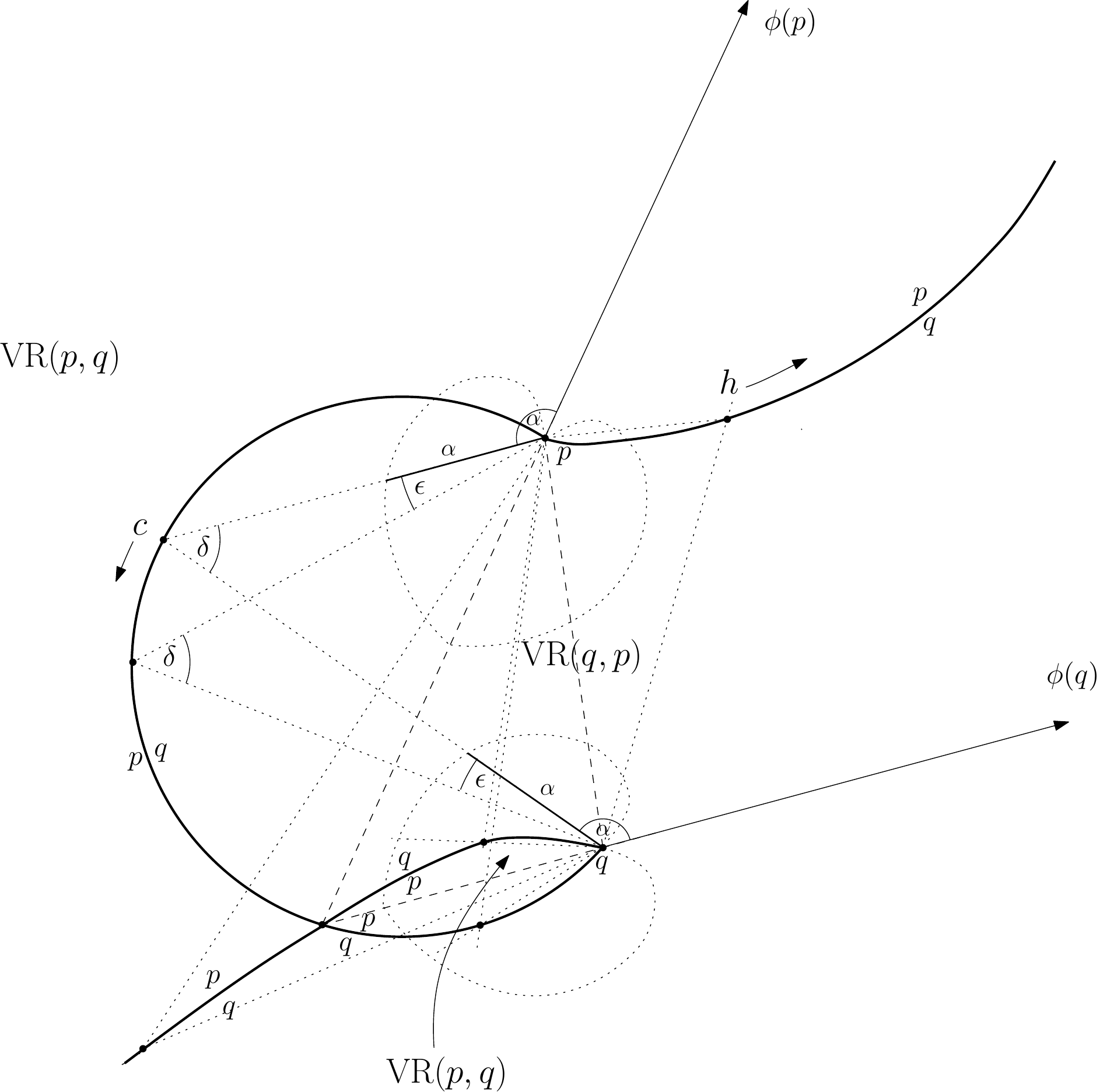}
\caption{A bisector of two points $p,q$ in the minimum angle distance. The Voronoi region $\mbox{VR}(p,q)$ of $p$ is disconnected.}
\label{sym-fig}
\end{center}
\end{figure}
De Berg et al.\ did not describe what shape bisectors have in this setting; they only observed that a bisector cannot intersect any line more than three times, and that it can divide the plane into up to four faces, never more. We observe that, as in the unsymmetric case, the bisector can contain a circular segment. Indeed, as bisector point $c$ moves downwards along the bisector, the sum of the angles at $p$ and $q$ in the triangle $p,c,q$ must remain constant.
But when point $h$ moves to the right, the \emph{difference} of the angles at $p$ and $q$ in the triangle $p,h,q$ must stay constant.
What curve results from this requirement?

\break

\section{The hyperbola as locus of constant angle difference}\label{hyper-sect}

We can prove the following.

\begin{theorem}
Let $P,Q,C$ be a triangle whose angles at vertices $P$ and $Q$ have  difference $\mu - \nu$. Then the locus of all points 
$H$ to which $C$ can be moved without changing $\mu-\nu$ is a hyperbola.
\end{theorem}
\begin{proof}
First, we rotate the triangle such that its edges $PC$ and $QC$ make angles $-\alpha$ and $\alpha$ with the $X$-axis, respectively; see Figure \ref{hyper-fig}. Then we translate the triangle to make the origin the center of $PQ$, and finally scale it so that $P=(h,1)$ and $Q=(-h,-1)$ for some $h$.

Now suppose $C$ moves to some point $H=(x,y)$ such that the angles at $P$ and $Q$ become larger %
than $\mu$ and $\nu$ 
by the same amount, $\epsilon$. Since the angles of $PH$ and $QH$ with the $X$-axis still have identical absolute values, we 
have
\[
    \frac{|HT|}{|PT|} = \frac{|SH|}{|QS|}.
\]
Substituting coordinates results in
\[
    \frac{1-y}{x-h} = \frac{1+y}{x+h},
\]
which solves to $y = \frac{h}{x}$, the equation of a hyperbola.
\end{proof}

\begin{figure}
\begin{center}
\includegraphics[scale=0.8]{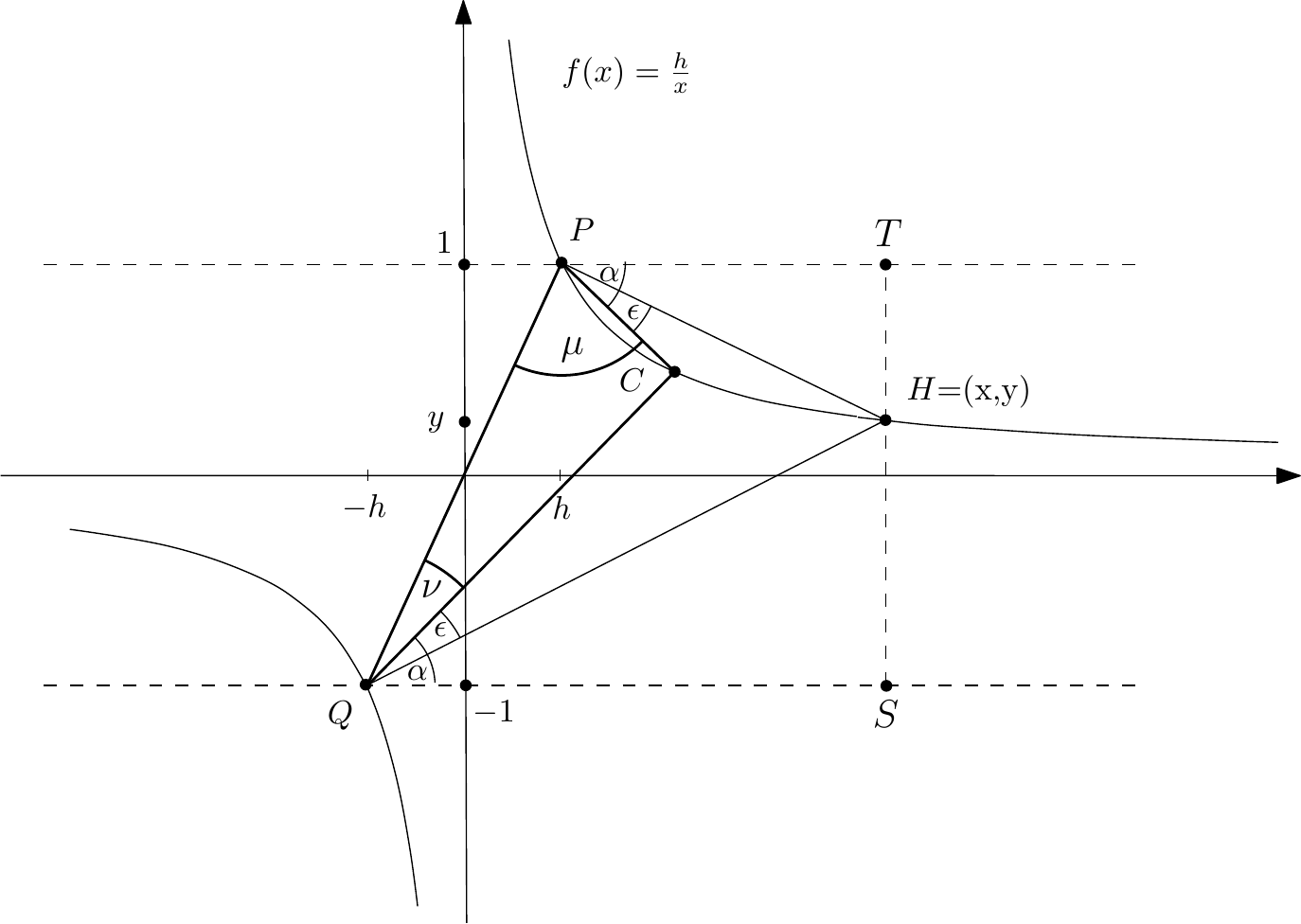}
\caption{A hyperbola maintaining constant angle difference at $P$ and $Q$.}
\label{hyper-fig}
\end{center}
\end{figure}
%


 \end{document}